\documentclass[11pt]{amsart}
\usepackage{graphicx} % Required for inserting images

\usepackage{array}
\usepackage[utf8]{inputenc}
\usepackage{hyperref}

\usepackage{amssymb}
\usepackage{pgfplots}
\usepackage{amsthm}
\usepackage{stmaryrd}
\usepackage{fullpage}
\usepackage{tikz-cd}
\usepackage{todonotes}
\usepackage{amsmath}
\usepackage{amsfonts}
\usepackage{subcaption}

\newtheorem{definition}{Definition}[section]
\newtheorem{theorem}[definition]{Theorem}
\newtheorem{lemma}[definition]{Lemma}
\newtheorem{corollary}[definition]{Corollary}
\newtheorem{proposition}[definition]{Proposition}

\newtheorem{example}{Example}
\newtheorem{remark}{Remark}

\newtheorem{construction}{Construction}

\newcommand{\calC}{\mathcal{C}}

\newcommand{\calF}{\mathcal{F}}
\newcommand{\calM}{\mathcal{M}}

\newcommand{\F}{\mathbb{F}}
\newcommand{\K}{\mathbb{K}}
\newcommand{\C}{\mathbb{C}}
\newcommand{\N}{\mathbb{N}}

\newcommand{\mun}[1]{\boldsymbol{\mu}_{#1}}

\newcommand{\Tr}{\operatorname{Tr}}
\newcommand{\Trqp}{\Tr_{\F_q/\F_p}}

\newcommand{\ev}[1]{\operatorname{ev}_{#1}}

\let\oldexample\example
\renewcommand{\example}{\oldexample\normalfont}

\let\oldremark\remark
\renewcommand{\remark}{\oldremark\normalfont}

\let\oldconstruction\construction
\renewcommand{\construction}{\oldconstruction\normalfont}

\title{Exact Cardinality and Nonredundant Parametrization of Character-Polynomial Codes}
\author{Felice Manganiello and Sitraka Randrianarivo}
\address{School Of Mathematical and Statistical Sciences, Clemson University, Clemson, SC 29634, USA}

\begin{document}

\maketitle

\begin{abstract}
Character-polynomial codes are constructed by evaluating finite field polynomials and mapping the results to complex roots of unity through additive characters. This paper shows that, over extension fields, the original polynomial family may contain redundancies: distinct polynomials can generate the same codeword. We identify the source of this non-injectivity through the trace map and cyclotomic cosets, determine the exact code cardinality, and construct a refined polynomial family that parametrizes the code without redundancy. These results give corrected parameters for CP codes and clarify their algebraic structure.
\end{abstract}

\section{Introduction}

Analog subspace codes were introduced in \cite{soleymani2021analog} as a coding framework for large
non-coherent wireless networks, where many relay nodes amplify and forward
signals without knowing the point-to-point channel gains or the structure of
the network. In such networks, conventional techniques such as channel
estimation, link-level block coding, and successive interference cancellation
do not scale well as the number of small cells, relays, and connected devices
increases. The main idea is to encode information not in individual transmitted
vectors, but in the subspace that those vectors
span over an analog field such
as $\mathbb{R}$ or $\mathbb{C}$. This leads to the notion of an analog
operator channel, which is an analog-domain counterpart of the
Koetter--Kschischang operator channel \cite{koetter2008coding}. In this model, rank deficiency in the
network is represented as subspace erasures, while interference from
neighboring cells is represented as subspace errors. Formally, the channel
input and output are subspaces, and the channel output can be written as
\[
    V = \mathcal{H}_k(U) \oplus E,
\]
where $\mathcal{H}_k(U)$ is a lower-dimensional subspace of the transmitted
subspace $U$, and $E$ is an error or interference space.

The authors of \cite{soleymani2021analog} use this analog operator channel to study communication over
non-coherent wireless networks through subspace codes. Since the transmitter
and receiver do not know the channel matrix, the information that can be
reliably conveyed is the subspace generated by the transmitted signal blocks,
rather than the exact transmitted matrix. A distance based on projection
matrices is defined on the set of subspaces, allowing minimum-distance decoding
to recover the transmitted subspace when the total number of subspace errors
and erasures is sufficiently small. As an explicit construction, the paper
defines character-polynomial (CP) codes. These are one-dimensional complex
Grassmannian codes obtained by evaluating polynomials over a finite field
$\mathbb{F}_q$ and then mapping the finite field values to complex roots of unity using an additive character $\chi$. More precisely, for a family of
polynomials $\mathcal{F}$ and nonzero evaluation points
$\alpha_i \in \mathbb{F}_q$, each polynomial $f \in \mathcal{F}$ defines a
complex vector
\[
    \bigl(\chi(f(\alpha_1)), \chi(f(\alpha_2)), \ldots,
    \chi(f(\alpha_m))\bigr),
\]
and the corresponding codeword is the one-dimensional subspace spanned by this
vector. The resulting CP codes use bounds on character sums to guarantee small
correlations between distinct codewords and therefore good minimum-distance
properties \cite{soleymani2021analog}.

Since their introduction, character-polynomial (CP) codes have mainly been
used as explicit algebraic constructions for Grassmannian packing problems and
as a bridge between finite field coding theory and analog subspace coding.
The classical work \cite{conway1996packing} on Grassmannian packings highlights the importance of the chordal metric for measuring subspace separation. CP codes fit naturally into
this framework, as their induced distances correspond to a scaled version of the
squared chordal metric. One follow-up direction extends the original one-dimensional CP code
construction to higher-dimensional subspaces in Grassmannian space. In this
setting, CP-type constructions are used to build new packings in complex and
real Grassmannians, and they also serve as inner codes in concatenated
Grassmann codes. This shows that the CP construction is not limited to lines in
$G_{1,m}(\mathbb{C})$, but can also be used to obtain higher-dimensional
Grassmann codes with competitive finite-length parameters and favorable
asymptotic behavior \cite{soleymani2021new}. A second line of work
studies CP codes from an algorithmic decoding perspective. Riasat and
Mahdavifar observe that, before applying the character map, a CP code can be
viewed as a structured subcode of a Reed--Solomon code. This connection allows
classical Reed--Solomon decoding tools, including minimum-distance decoding and
list decoding, to be adapted to CP subspace codes \cite{riasat2024decoding}.

More recently, CP codes have also been used beyond the original analog
operator-channel setting, especially as structured codebooks for Grassmannian
quantization in wireless communication. In limited-feedback MISO/MIMO 
precoding, the receiver can quantize channel-state information to a CP codeword
and send the corresponding index to the transmitter. The CP codebook structure
gives a tunable tradeoff between the number of feedback bits and the resulting
beamforming gain, while the equal-magnitude entries of codewords are useful
for equal-gain or per-antenna power constrained transmission
\cite{gooty2025precoding}.
The relevance of Grassmannian packings in this setting goes back to the classical work \cite{love2003grassmannian}, which established Grassmannian line packings as a
natural framework for designing limited--feedback beamforming codebooks.
A related direction studies covering rather
than packing: recent work on Hamming and Grassmann spaces uses the
Reed--Solomon/character viewpoint underlying CP codes to define
character--Reed--Solomon (CRS) codes for Grassmannian quantization and to
analyze their average covering radius \cite{riasat2025covering}. These
follow-up works suggest that CP codes have developed from a construction for
analog subspace coding into a broader algebraic tool for Grassmannian packing,
decoding, precoding, and quantization.

In this paper, we study the
cardinality of CP codes in detail.

\subsection{Main Contribution}

The main contribution of this paper is a precise characterization of the
cardinality of character-polynomial (CP) codes, correcting a
non-injectivity issue in the original construction. 

A character-polynomial (CP) code is a collection of one-dimensional
subspaces of $\mathbb{C}^n$ obtained by evaluating polynomials over a
finite field and then mapping the evaluation values to the complex roots of unity through a fixed nontrivial additive character. More precisely, for a
family of polynomials \[
\mathcal{F}
:=
\left\{
    f(x)=
    \sum_{\substack{i\in\{1,\dots,k-1\}\\ p\: \nmid \: i}}
    f_i x^i
    \ \middle|\ 
    f_i\in\mathbb{F}_q \text{ for all } i
\right\},
\] distinct
evaluation points $\alpha_1,\ldots,\alpha_m\in\mathbb{F}_q$, and a fixed
nontrivial additive character $\chi:\mathbb{F}_q\to\mathbb{C}$, each
polynomial $f\in\mathcal{F}$ gives the complex vector
\[
    \bigl(\chi(f(\alpha_1)),\ldots,\chi(f(\alpha_m))\bigr)\in\mathbb{C}^m.
\]
The CP code is the set of one-dimensional subspaces spanned by these
vectors.
In
\cite{soleymani2021analog}, the cardinality of a CP code is identified with
the cardinality of the polynomial family $\mathcal{F}$. 
The exclusion of monomials with degree divisible by the field characteristic
$p$ is motivated by the fact that, if all polynomials are allowed, then
different polynomials may generate the same codeword. This removes one obvious
source of redundancy in the construction.

Nevertheless, we prove that this condition alone does not eliminate all
redundancies. In particular, there may still exist distinct polynomials
$f,g\in\mathcal{F}$ whose evaluations, after applying the additive
character, span the same one-dimensional subspace of $\mathbb{C}^n$.
Therefore, the map used for defining the CP code construction is not necessarily injective. It follows
that the actual cardinality of the resulting CP code can be strictly smaller
than $|\mathcal{F}|$. 
The loss in cardinality is ultimately caused by the fact that
additive characters of $\mathbb{F}_q$ do not distinguish all
elements of $\mathbb{F}_q$ when $q$ is not prime. More precisely,
let $q=p^r$, where $p$ is the characteristic of $\mathbb{F}_q$, and let $\mu_p$ represent the set of $p$-th roots of unity.
Every additive character of $\mathbb{F}_q$ can be written in the
form
\[
    \chi(x)
    =
    \chi_{p}\bigl(\Tr_{\mathbb{F}_q/\mathbb{F}_p}(ax)\bigr)
\]
for some $a\in\mathbb{F}_q $, where $\chi_p:\mathbb{F}_p\to \C$ is a nontrivial additive
character of the prime field and
$\Tr_{\mathbb{F}_q/\mathbb{F}_p}$ is the field trace function. This
factorization can be represented by the commutative diagram
\begin{figure}[ht]
\centering
\begin{tikzpicture}[
    >=Stealth,
    node distance=2.2cm and 2.8cm,
    every node/.style={font=\normalsize},
    arrow/.style={->, thick}
]
    \node (Fq) at (0,0) {$\F_q$};
    \node (C)  at (5,0) {$\C$};
    \node (Fp) at (2.5,-1.8) {$\F_p$};

    \draw[arrow] (Fq) -- node[above] {$\chi$} (C);

    \draw[arrow] (Fq) -- 
        node[below left] {$x \mapsto \Trqp(ax)$} 
        (Fp);

    \draw[arrow] (Fp) -- 
        node[below right] {$\chi_p$} 
        (C);
\end{tikzpicture}
\caption{The character $\chi$ as the composition of the trace map with $\chi_p$.}
\label{fig:trace-character}
\end{figure}

Thus, since the trace map is onto $\F_p$ when $a\neq 0$, the image of $\chi$ has cardinality
$p$, while its domain has cardinality $q=p^r$. In particular,
for $r>1$, the character $\chi$ is not injective. Consequently, different
finite field evaluations may be mapped to the same complex value on
the unit circle. This non-injectivity induces possible collisions
among codewords: two distinct polynomials may give evaluation
vectors spanning the same one-dimensional complex subspace.
Therefore, the number of distinct codewords can be strictly
smaller than the number of polynomials used to define them.

\begin{example}\label{ex:strict_cardinality}

Let $p=3$, $q=9$, $m=9$ and $k=8$. Consider the additive character
\[
    \chi(\ell)
    =
    e^{\left(
        \frac{2\pi i}{3}
        \Tr_{\F_9/\F_3}(\ell)
    \right)},
    \qquad \ell\in\F_9.
\]
Let $\alpha\in\F_9$ be a primitive element with minimal polynomial
$x^2+2x+2$ over $\F_3$. A direct computation shows that
\[
    \Tr_{\F_9/\F_3}\!\bigl((1+\alpha)\ell^4\bigr)=0,
    \qquad \forall\,\ell\in\F_9.
\]
Therefore, $    \chi\!\bigl((1+\alpha)\ell^4\bigr)=1 $ for $\ell\in\F_9.$
It follows that the polynomial
    $f(x)=(1+\alpha)x^4\in\mathcal F$
produces the same evaluation vector as the zero polynomial
$g(x)=0$, namely $(1,\ldots,1)$. Consequently, $f$ and $g$
define the same codeword 
    $\langle (1,\ldots,1)\rangle$.

A second example is provided by the polynomials 
    $f(x)=x^5$, and $
    g(x)=x^7$. Indeed, one verifies that
    $\chi(\ell^5)=\chi(\ell^7)$ for $\ell\in\F_9$.
Hence, the corresponding evaluation vectors coincide, and therefore
$f$ and $g$ determine the same codeword.

\end{example}

Figure \ref{fig:foobar} shows that the cardinality predicted by
$|\mathcal{F}|$ can overestimate the true size of the resulting CP code, given in Proposition \ref{CP:size&poly}, 
highlighting the need for an exact characterization of the code
cardinality.

\begin{figure}[htb]
    \centering

    \begin{subfigure}[t]{0.45\textwidth}
        \vspace{0pt}
        \centering
        \begin{tabular}{c|c|c}
        \textbf{$r=\log_2 q$} & \textbf{$\log_2|\calF|$} & $\log_{2}|\calC(\calF)|$ \\
        \hline
        1  & 1    & 1    \\
        2  & 4    & 3    \\
        3  & 12   & 7    \\
        4  & 32   & 15   \\
        5  & 80   & 31   \\
        6  & 192  & 63   \\
        7  & 448  & 127  \\
        8  & 1024 & 255  \\
        9  & 2304 & 511  \\
        10 & 5120 & 1023 \\
        \end{tabular}
        \caption{}
    \end{subfigure}
    \hfill
    \begin{subfigure}[t]{0.50\textwidth}
        \vspace{0pt}
        \centering
        \begin{tikzpicture}
        \begin{axis}[
            xlabel={Field size},
            width=\linewidth,
            height=0.72\linewidth,
            xmode=log,
            log basis x=2,
            legend pos=north west,
            ymajorgrids=true,
            grid style=dashed,
            xtick={2, 4, 8, 16, 32, 64, 128, 256, 512, 1024},
            xticklabels={$2^1$,$2^2$,$2^3$,$2^4$,$2^5$,$2^6$,$2^7$,$2^8$,$2^9$,$2^{10}$},
            ytick={0,1000,2000,3000,4000,5000,6000},
yticklabels={$0$, $10^3$, $2\!\cdot\!10^3$, $3\!\cdot\!10^3$,
             $4\!\cdot\!10^3$, $5\!\cdot\!10^3$, $6\!\cdot\!10^3$},
        ]

        \addplot[
            color=blue,
            mark=square,
        ]
        coordinates {
            (2, 1)
            (4, 4)
            (8, 12)
            (16, 32)
            (32, 80)
            (64, 192)
            (128, 448)
            (256, 1024)
            (512, 2304)
            (1024, 5120)
        };

        \addplot[
            color=red,
            mark=square,
        ]
        coordinates {
            (2, 1)
            (4, 3)
            (8, 7)
            (16, 15)
            (32, 31)
            (64, 63)
            (128, 127)
            (256, 255)
            (512, 511)
            (1024, 1023)
        };

        \legend{$\log_2|\calF|$, $\log_2|\calC(\calF)|$}
        \end{axis}
        \end{tikzpicture}
        \caption{}
    \end{subfigure}

    \medskip

    \begin{subfigure}[t]{0.45\textwidth}
        \vspace{0pt}
        \centering
        \begin{tabular}{c|c|c}
        \textbf{$r=\log_3 q$} & $\log_3|\calF|$ & $\log_{3}|\calC(\calF)|$ \\
        \hline
        1  & 2      & 2      \\
        2  & 12     & 8      \\
        3  & 54     & 26     \\
        4  & 216    & 80     \\
        5  & 810    & 242    \\
        6  & 2916   & 728    \\
        7  & 10206  & 2186   \\
        8  & 34992  & 6560   \\
        9  & 118098 & 19682  \\
        10 & 393660 & 59048  \\
        \end{tabular}
        \caption{}
    \end{subfigure}
    \hfill
    \begin{subfigure}[t]{0.50\textwidth}
        \vspace{0pt}
        \centering
        \begin{tikzpicture}
        \begin{axis}[
            xlabel={Field size},
            scaled y ticks=false,
            width=\linewidth,
            height=0.72\linewidth,
            xmode=log,
            log basis x=3,
            legend pos=north west,
            ymajorgrids=true,
            grid style=dashed,
            xtick={3, 9, 27, 81, 243, 729, 2187, 6561, 19683, 59049},
            xticklabels={$3^1$,$3^2$,$3^3$,$3^4$,$3^5$,$3^6$,$3^7$,$3^8$,$3^9$,$3^{10}$},
            ytick={0,100000,200000,300000,400000},
yticklabels={$0$, $10^5$, $2\!\cdot\!10^5$, $3\!\cdot\!10^5$,
             $4\!\cdot\!10^5$},
        ]

        \addplot[
            color=blue,
            mark=square,
        ]
        coordinates {
            (3, 2)
            (9, 12)
            (27, 54)
            (81, 216)
            (243, 810)
            (729, 2916)
            (2187, 10206)
            (6561, 34992)
            (19683, 118098)
            (59049, 393660)
        };

        \addplot[
            color=red,
            mark=square,
        ]
        coordinates {
            (3, 2)
            (9, 8)
            (27, 26)
            (81, 80)
            (243, 242)
            (729, 728)
            (2187, 2186)
            (6561, 6560)
            (19683, 19682)
            (59049, 59048)
        };

        \legend{$\log_3|\calF|$, $\log_3|\calC(\calF)|$}
        \end{axis}
        \end{tikzpicture}
        \caption{}
    \end{subfigure}

    \caption{The effect of this cardinality loss can be seen in these plots
where the size of the defining polynomial family $\mathcal{F}$ is
compared with the actual cardinality of the CP code. We consider
the cases $p=2$ and $p=3$, take $q=p^r$ for $r=1,\ldots,10$,
and set $k=m=q$.  Subfigures (a) and (b) correspond to $p=2$; subfigures (c) and (d) correspond to $p=3$.}
    \label{fig:foobar}
\end{figure}
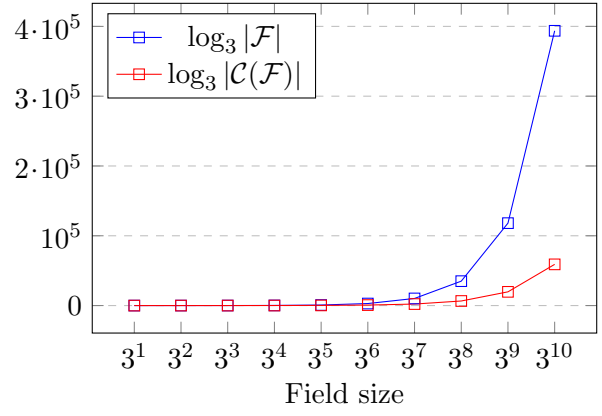

After reviewing the necessary preliminaries in
Section~\ref{sec:prelim}, we study the cardinality of CP codes.
Since the loss of cardinality originates from the interaction between
additive characters and the trace map, Section~\ref{sec:cardinality}
determines the cardinality of the image of the trace map. This analysis is
then used in Section~\ref{sec:trace-polynomials} to characterize the set of
polynomials that are uniquely represented through their trace evaluations.
Finally, in Section~\ref{sec:CP-polynomials}, we exploit this characterization
to construct a refined family of defining polynomials that yields a one-to-one
correspondence between polynomials and codewords, thereby providing a unique
parametrization of the code.

\subsection{Result Impact}

The refined cardinality formula also clarifies how CP code parameters should
be interpreted in later applications of the construction. Since the rate of a
subspace code depends on the number of distinct subspaces rather than on
the number of polynomials used to generate them, any use of CP codes over
extension fields should distinguish between the defining polynomial family and
the resulting set of codewords. This distinction primarily affects
cardinality-dependent quantities, such as rate and rate--distance
comparisons.

This observation is relevant, for example, to the CP-based Grassmannian
packing constructions in \cite{soleymani2021new}. For parameter
choices over non-prime fields, the size of the resulting packing should be
computed as the number of distinct subspaces produced by the construction. A similar remark applies in the higher-dimensional variants when compared with \cite{calderbank1999group}. In decoding applications such as
\cite{riasat2024decoding}, the same issue can be addressed by fixing a
refined polynomial family that parametrizes codewords uniquely; then
decoding polynomials and decoding codewords become equivalent again. Under
parameter regimes where the original and refined polynomial families coincide,
the existing simulations and conclusions are unaffected. Finally, works that
mainly use CP codes over prime fields, such as
\cite{gooty2025precoding}, are not expected to be impacted by the
extension-field non-injectivity considered here.

\section{Preliminaries}\label{sec:prelim}

\subsection{Additive Characters}
In this paper, $\F_q$ denotes the finite field of size $q=p^r$, a prime power for some $r \in \mathbb{N}$. Formally, the field trace function from $\F_q$ to $\F_p$ is
\[\Trqp(a)=\sum_{j=0}^{r-1}a^{p^j}.\]
The trace function, from $\F_q$ to $\F_p$, is a surjective $\F_p$-linear map, meaning that for any $a,b\in\F_q$ and $\lambda_1,\lambda_2\in\F_p$
 \[\Trqp(\lambda_1 a+\lambda_2 b)=\lambda_1\Trqp(a)+\lambda_2\Trqp(b).\]

Let us now introduce the notion of character of a finite abelian group. For background on characters of
finite fields, see Chapter 5 in \cite{lidl1997finite} for more details on the topic. For a given $n \in \N$, 
\[\mun{n}=\left\{\left.\omega^l=e^{\frac{2i\pi}{n}l}\;\right|\; l=0,\dots,n-1\right\}\] 
denotes the set of the $n$-th complex roots of unity. 

\begin{definition}
    An additive character of $\F_q$ is a group homomorphism $\chi:\F_q \to \mun{p}$ from $\F_q$ as an additive group into the
multiplicative group of complex roots of unity. 

\end{definition} 

The following theorem gives a complete characterization of the additive characters of a finite field.

\begin{theorem}[Theorem 5.7 \cite{lidl1997finite}] \label{add_char_def}
    If $\chi:\F_q \to \mun{p}$ is an additive character of $\F_q$, then there exists $a \in \F_q$ such that $$\chi(\ell)=e^ {\left( \frac{2i\pi}{p} \Trqp (a\ell) \right)}.
    $$
    We say that $\chi$ is nontrivial if $a \neq 0$.
\end{theorem}

This theorem is represented graphically with the commutative diagram in Figure \ref{fig:trace-character}.

\subsection{CP Codes}

Character-polynomial are originated from finite-field evaluation codes.

\begin{definition}[\cite{soleymani2021analog}] \label{def:cp_code}
    Let $q=p^r$, $\chi$ a nontrivial additive character of $\F_{q}$,  $k\in \N$ with $k\leq q$, and
    $$\calF:=\left\{ f=\sum_{\substack{i\in \{1,\dots,k-1\}\\p\:\nmid\: i}} f_ix^i \; \middle | \; f_i \in \F_q\right\}.$$
    If $\F_q=\{a_1,\dots,a_q\}$, a Character Polynomial (CP) code is  
    $$\calC(\calF) = \left\{\langle(\chi(f(a_1)),\dots,\chi(f(a_q)))\rangle_\C\mid f\in \calF\right\}.$$
\end{definition}

Throughout the rest of the paper, we use the canonical nontrivial character
$\chi(x)=\chi_p\!\left(\operatorname{Tr}_{\mathbb{F}_q/\mathbb{F}_p}(x)\right).$ This assumption entails no loss of generality for either the cardinality of the
code or the polynomial classification problem.

Let $G_{1,m}(\C)$ denote the set of one-dimensional subspaces of $\C^m$ and 
\begin{align*}
    \ev{\Trqp}: \F_q[x]^{<k} & \longrightarrow   \F_p^q \\
            f & \longmapsto (\Trqp(f(a_i)))_{i=0,\dots,q-1},
\end{align*}
the trace evaluation with $a_i=\alpha^i$ for $i=0,\dots,q-2$ with $\alpha \in \F_q$ a primitive element and $a_{q-1}=0$.
A CP code is obtained from the
following composition,
\begin{figure}[htb]
\centering
\begin{tikzpicture}[
    >=Stealth,
    every node/.style={font=\normalsize},
    arrow/.style={->}
]
    \node (A) at (0,0) {$\calF$};
    \node (B) at (3.2,0) {$\F_p^q$};
    \node (C) at (5.9,0) {$\mun{p}^q$};
    \node (D) at (9.0,0) {$G_{1,q}(\C)$};

    \draw[arrow] (A) -- node[above] {$\ev{\Trqp}$} (B);
    \draw[arrow] (B) -- node[above] {$\chi_p$} (C);
    \draw[arrow] (C) -- node[above] {$\langle \cdot \rangle_\C$} (D);
\end{tikzpicture}
\caption{The composition from  polynomials over $\F_q$ to points of $G_{1,m}(\C)$.}
\label{fig:trace-evaluation-grassmannian}
\end{figure}
where $\langle \cdot \rangle_\C$ denotes the span operator and, by abuse of notation, $\chi_p$ also denotes the map obtained by applying the additive character $\chi_p:\F_p \to \mu_p$ entrywise to vectors in $\F_p^q$.

\subsection{Cyclotomic Cosets}
In Section \ref{sec:cardinality}  we work with cyclotomic cosets. We review their definition and basic properties here. 

\begin{definition}
    For $l \in \{ 0,\dots,q-2 \}$, the $p$-cyclotomic coset modulo $(q-1)$ of $l$ is 
$$C(l):=\{ lp^i \mod (q-1) \mid i=0,\dots,r-1\}.$$
Moreover, we define $C(q-1):=\{ q-1 \}$.
\end{definition}

Note that the $p$-cyclotomic cosets modulo (q-1) are the equivalence classes of the following equivalence relation
\[l\sim j \Longleftrightarrow j=lp^i\mod (q-1)\] for some $i$. Moreover, $C(q-1)$ is not meaningful a priori, since $p$-cyclotomic cosets are defined modulo $(q-1)$, we therefore extend the definition to $C(q-1)$. 

The following holds.

\begin{lemma}
    
    For $l,j \in \{ 0,\dots,q-2 \}$, we have \[C(l)=C(j) \iff l \in C(j).\] Also, the sets $C(l)$ for $l\in\{0,\dots,q-2\}$ form a partition of $\mathbb{Z}_{q-1}$.
\end{lemma}

\begin{example} \label{cyclo_example}
    Let $p=3$ and $q=9$. We have that: $C(0)=\{0\}$, $C(1)=C(3)=\{ 1,3 \}$, $C(2)=C(6)=\{ 2,6 \}$, $C(4)=\{4\}$, and $C(5)=C(7)=\{ 5,7 \}$.
\end{example}

Subfields of $\F_q$ arise from cyclotomic cosets.

\begin{lemma} \label{cyclo_field}
    Let $\alpha \in \F_q$ be a primitive element and $l \in \{ 0,\dots,q-1 \}$. We have
    $$\K_{l} = \langle  (\alpha^l)^i \mid i=0,\dots,q-2 \rangle_{\F_p} \subseteq \F_{q}$$
    is a subfield with cardinality $p^{|C(l)|}$. Moreover, if $j \in C(l)$ we have
    $$\K_l=\K_j.$$
\end{lemma}

\begin{proof}
The set $\K_l$ contains $1$, is finite and is closed under addition, so it remains to prove that it is closed under multiplication to show that it is a field. Let $\sum_{i=0}^{q-2} \lambda_{i,1}(\alpha^l)^i \in \K_l$ and $\sum_{s=0}^{q-2} \lambda_{s,2} (\alpha^l)^s \in \K_l$, we have
    $$
    \left( \sum_{i=0}^{q-2} \lambda_{i,1}(\alpha^l)^i \right) \left( \sum_{s=0}^{q-2} \lambda_{s,2} (\alpha^l)^s \right) = \sum_{n=0}^{q-2} \left(\sum_{i+s=n \pmod{q-1}} \lambda_{i,1} \lambda_{s,2} \right) (\alpha^l)^n \in \K_l
    .$$
    Now let $j \in C(l)$, then, for some $t \in \{ 0,\dots,r-1 \}$,  
    $j = lp^t \mod (q-1)$ and \[\alpha^j \in \{ (\alpha^l)^i \mid i=0,\dots,q-2 \}.\]  Since this generating set is closed under multiplication, we have 
    $$\{ (\alpha^j)^i \mid i=0,\dots,q-2 \} \subset \{ (\alpha^l)^i \mid i=0,\dots,q-2 \} \implies \K_j \subseteq \K_l.$$
    The other direction comes from $l \in C(j)$ by the same argument. Hence $\K_l=\K_j$.
    
    It remains to prove that 
    $$|\K_l| = p^{|C(l)|}.$$    Let $m_{\alpha^l}$ be the minimal polynomial of $\alpha^l$ over $\F_p$, then it holds that
    $$
    \dim_{\F_p} \langle (\alpha^l)^j \mid j=0,\dots,q-2 \rangle_{\F_p} = \deg(m_{\alpha^l}).
    $$
    For $l \neq q-1$, if $d$ is minimal such that $(\alpha^l)^{p^d} = \alpha^l$, we have
    $$
    m_{\alpha^l} = \prod_{j=0}^{d-1} (x-\alpha^{lp^j}) = \prod_{i \in C(l)} (x-\alpha^i) \implies \deg(m_{\alpha^l})=|C(l)|.
    $$
    If $l=q-1$, then 
    $$m_{\alpha^{q-1}}=x-1 \implies \deg(m_{\alpha^{q-1}})=1=|C(q-1)|.$$
    Therefore, we have 
    $$\operatorname{dim}_{\F_p} \langle (\alpha^l)^i \mid i=0,\dots,q-2 \rangle_{\F_p} = |C(l)|,$$
    and
    $$|\K_l| = p^{|C(l)|}.$$

\end{proof}

\begin{example}
    Let $p=3$ and $q=9$. 
    \begin{itemize}
        \item If $l=2$, we have
        $\K_2 = \langle 1,\alpha^2,\alpha^4,\alpha^6 \rangle_{\F_3} = \F_9.$
        \item If $l=4$, we have
        $\K_4 = \langle 1,\alpha^4 \rangle_{\F_3}= \langle1,2\rangle_{\F_3} = \F_3.$
        \item If $l=6$ we have
        $\K_6 = \langle 1,\alpha^6,\alpha^4,\alpha^2 \rangle_{\F_3} = \F_9 = \K_2.$
    \end{itemize}
\end{example}

We now prove one of the key
ingredients of this work that is used throughout the manuscript, which is a relation between the evaluation of $\Trqp(ax^l)$ on $\F_q$ and the evaluation of $\Trqp(ax)$ on $\K_l$ for some $a \in \F_q$.

\begin{proposition}\label{prop:tr_zero}
    Let $a \in \F_q$ and $l \in \{ 0,\dots,q-1 \}$. We have, 
    $$\Trqp(ab^l)=0,\ \forall b \in \F_q \iff \Trqp(ac)=0,\ \forall c \in \K_l$$
\end{proposition}

This proof is a direct consequence of the definition of $\K_l$ and the $\F_p$-linearity of the $\Trqp$.

\begin{proof}
    The case $l=0$ is immediate. Now for $l \in \{ 1,\dots,q-1 \}$, suppose $\Trqp(ab^l)=0$ for all $b\in\F_q$. This means, if $\alpha \in \F_q$ is primitive, we have
    $$\Trqp(a(\alpha^i)^l)=0,\ \forall i=0,\dots,q-2 \iff \Trqp(ab)=0,\ \forall b \in \{ (\alpha^i)^l \mid i=0,\dots,q-2\}.$$
    By the $\F_p$-linearity of $\Trqp$ we have 
    $$ \Trqp(ac)=0,\ \forall c \in \langle\{ (\alpha^i)^l \mid i=0,\dots,q-2\}\rangle_{\F_p}=\K_l.$$
    Now, if $\Trqp(ac)=0$ for all $c\in \K_l$, since $\{ b^l \mid b\in \F_q\}\subseteq \K_l$,  then 
    $\Trqp(ab^l)=0$ for all $b\in \F_q$.
    \end{proof}

\section{Cardinality of the Trace Evaluation Map} \label{sec:cardinality}

    In this section, we focus on the image cardinality of the composition 
    \[ \begin{tikzcd}[ampersand replacement=\&, column sep=large, row sep=large] \F_q[x]^{<k} \ar[r, "\ev{\Trqp}"] \& \F_p^q \ar[r, "\chi_p"] \& \mun{p}^q  \end{tikzcd}. \]
    This composition differs from Figure \ref{fig:trace-evaluation-grassmannian} in that we allow all polynomials of degree
less than $k$, rather than only the
defining family of CP codes. 
    We leave the last part of the composition to Section \ref{sec:CP-polynomials}.
    Since $\chi_p:\F_p^q \to \mun{p}^q$ is bijective, we focus on the image of $\ev{\Trqp}$. 
    
    We begin by studying the image of monomials of a fixed degree and describing the structure of the subspaces they generate. More precisely, we compute the
dimension of each image as an $\F_p$-vector space. We then consider collections of degrees 
    organized according to their $p$-cyclotomic cosets modulo $(q-1)$, and we show that the corresponding families of monomials have linearly independent images and that the image of their sum decomposes as the direct sum of the individual images. Finally, combining these two ingredients, we derive an explicit formula for the dimension of $\ev{\Trqp}(\F_q[x]^{<k})$.
    
    We need the following definition.
    
\begin{definition}
    For $l \in \{ 0,\dots,q-1 \}$ denote by $\calM_l$ the set of all monomials of degree $l$ including the zero polynomial, meaning
    \[\calM_l=\{ax^l \mid a \in \F_q\}.\]
\end{definition}
We include the zero polynomial so that $\calM_l$ is  a subspace of $\F_q[x]$. Let us first give a theorem characterizing all $\F_p$-linear transformations from $\F_q$ to $\F_p$.

\begin{theorem}[Theorem 2.24 \cite{lidl1997finite}] \label{isom_Tr}
    The map 
    \begin{align*}
           \phi:    \F_q & \longrightarrow \operatorname{Hom}_{\F_p}(\F_q,\F_p)\\
                     a   & \longmapsto  \phi_a 
        \end{align*}
    is an isomorphism of $\F_p$-vector spaces, where $\phi_a(x)=\Trqp(ax)$ for all $x \in \F_q$.
\end{theorem}

We can prove the following properties.

\begin{proposition} \label{size}
    The following holds:
    \begin{enumerate}
        \item $\dim_{\F_p}\ev{\Trqp}(\calM_l)=|C(l)|$ for $l\in\{0,\dots,q-1\}$\label{dimTM_l}
        \item $\ev{\Trqp}(\calM_j)=\ev{\Trqp}(\calM_l)$ if and only if $j\in C(l)$. \label{TM_l=TM_j}

    \end{enumerate}
\end{proposition}

\begin{proof}

    For $l=0$, $\ev{\Trqp}(\calM_0)=\langle(1,\dots,1)\rangle$, so
    $$\dim_{\F_p}\ev{\Trqp}(\calM_0)=1=|C(0)|.$$
    Now, for $l \in \{ 1,\dots,q-1 \}$, since $\Trqp(ax^l)=0$ if $x=0$ we only need to investigate their evaluation in $\F_q^*$. For this purpose, let
    \begin{align*}
        \ev{\Trqp}^*: \F_q[x]^{<k} & \longrightarrow   \F_p^{q-1} \\
            f & \longmapsto (\Trqp(f(a_i)))_{0\leq i\leq q-2}
    \end{align*}
    We have
    \begin{align*}
    \dim_{\F_p} \ev{\Trqp}(\calM_l) &= \dim_{\F_p} \ev{\Trqp}^*(\calM_l)=\dim_{\F_p} \left \{ \left(\Trqp (a(\alpha^j)^l)\right)_{0\leq j\leq q-2} \mid\ a \in \F_q \right \}.
    \\
    &\overset{1}{=}\dim_{\F_p} \{ (\phi_{\alpha^{lj}}(a))_{0\leq j\leq q-2} \mid\ a \in \F_q  \} \overset{2}{=} \dim_{\F_p} \left\langle (\phi_{\alpha^{lj}}(a))_{a \in \F_q} \mid j=0,\dots, q-2 \right\rangle_{\F_p}\\&=\dim_{\F_p} \left\langle \phi_{\alpha^{lj}} \mid j=0,\dots, q-2\right\rangle_{\F_p}\overset{3}{=} \dim_{\F_p} \langle (\alpha^l)^j \mid j=0,\dots,q-2 \rangle_{\F_p}
    \end{align*}
    where Equality 1 follows from Theorem \ref{isom_Tr}, Equality 2 follows from the rank-preserving property of the transpose operator, and Equality 3 follows from the isomorphism in Theorem \ref{isom_Tr}.

    Let $m_{\alpha^l}$ be the minimal polynomial of $\alpha^l$ over $\F_p$, then for $l \neq q-1$, it holds that
    $$
    \dim_{\F_p} \langle (\alpha^l)^j \mid j=0,\dots,q-2 \rangle_{\F_p} = \deg(m_{\alpha^l})=|C(l)|.
    $$

    Therefore, we have 
    $$
    \dim_{\F_p} \ev{\Trqp}(\calM_l) = \deg m_{\alpha^l} = |C(l)|.
    $$
    For $l=q-1$, we have $\alpha^l=1$ and $\ev{\Trqp}(\calM_{q-1})=\langle(1,\dots,1,0)\rangle_{\F_p}$. Therefore 
    $$\dim_{\F_p} \ev{\Trqp}(\calM_{q-1})= 1=|C(q-1)|.$$
    
    We now prove the second assertion. Let $j \in C(l)$, then  $j=lp^i \mod (q-1)$ for some  $i \in \{ 0,\dots,r-2 \}$. For any $a\in \F_q$
    $$\Trqp(ab^l)=\Trqp\left( (ab^l)^{p^i} \right)=\Trqp(a^{p^i}b^j)$$
    for $b \in \F_q$ implying that
    $$\ev{\Trqp}(\calM_l) \subseteq \ev{\Trqp}(\calM_j).$$
    Since 
    $$|\ev{\Trqp}(\calM_l)|=p^{|C(l)|}=p^{|C(j)|}=|\ev{\Trqp}(\calM_j)|$$
    we have 
    $$\ev{\Trqp}(\calM_l)=\ev{\Trqp}(\calM_j).$$
    Conversely, suppose that if $\ev{\Trqp}(\calM_l)=\ev{\Trqp}(\calM_j)$ then $j\in C(l)$. Since $\ev{\Trqp}(\calM_l)=\ev{\Trqp}(\calM_j)$,  for any $a\in \F_q$ there exists $b \in \F_q$ such that 
    $$\ev{\Trqp}(ax^l) = \ev{\Trqp}(bx^j).$$
    Consider $a,b\in \F_q$ such that 
    \begin{equation}\label{eq:evnot_0}
        0\neq \ev{\Trqp}(ax^l)= \ev{\Trqp}(bx^j).
    \end{equation}
    and let $\alpha \in \F_q$ be a primitive element, then we have that for $i\in \{ 0,\dots,q-2 \}$
    \begin{equation}\label{eq:tr_j=tr_l}
        \Trqp(a(\alpha^i)^l) = \Trqp(b(\alpha^i)^j).
    \end{equation}
    Let $m_{\alpha^l} \in \F_p[x]$ be the minimal polynomial of $\alpha^l$ over $\F_p$, and define
    $$g_{l,i} = x^im_{\alpha^l} \mod (x^{q-1}-1).$$
    It holds that $g_{l,i}(\alpha^l)=0$. So, for all $i \in \{0,\dots,q-2\}$, if $g_{l.i}=\sum_{s=0}^{q-2} g_sx^s$, by the $\F_p$-linearity of $\Trqp$, we have
    $$
    0 = \Trqp(ag_{l,i}(\alpha^l)) = \Trqp\left(a \sum_{s=0}^{q-2}g_s(\alpha^l)^s\right)=\sum_{s=0}^{q-2}g_s \Trqp(a(\alpha^l)^s).
    $$
    In addition, with Equation \eqref{eq:tr_j=tr_l}, we have
    $$
    \sum_{s=0}^{q-2} g_s\Trqp(a(\alpha^l)^s)=\sum_{s=0}^{q-2}g_s\Trqp(b(\alpha^j)^s)=\Trqp\left(b\sum_{s=0}^{q-2}g_s(\alpha^j)^s\right)=\Trqp(bg_{l,i}(\alpha^j)).
    $$
    Therefore, we have
    $$
    0 = \Trqp(bg_{l,i}(\alpha^j)) = \Trqp(b(\alpha^j)^im_{\alpha^l}(\alpha^j)) ,
    $$
    Since $i \in \{ 0,\dots,q-2 \}$ is arbitrary, we have
    $$
    \Trqp(b(\alpha^j)^im_{\alpha^l}(\alpha^j) )=0,\  \forall i \in \{0,\dots, q-2 \}.
    $$
    This implies that $\Trqp(bm_{\alpha^l}(\alpha^j)x)$ is identically zero in $\K_j $. Since $\K_j$ is a field by Lemma \ref{cyclo_field}, and $m_{\alpha^l}(\alpha^j)\in \K_j$ is nonzero  then multiplying $\K_j$ by $m_{\alpha^l}(\alpha^j)$ is a bijection, implying that $\Trqp(bx)$ is also identically zero in $\K_j$. By Proposition \ref{prop:tr_zero}, this implies $\Trqp(bx^j)$ is identically zero in $\F_q$, which is a contradiction with Equation \eqref{eq:evnot_0}. So we have 
    $$m_{\alpha^l}(\alpha^j)=0 \implies m_{\alpha^j} \mid m_{\alpha^l}.$$
    Since $m_{\alpha^l}$ is irreducible in $\F_p[x]$, we have 
    $$m_{\alpha^j} = m_{\alpha^l} \implies j \in C(l).$$
\end{proof}

\begin{example}\label{ex:overest_M_4}
    Let  $p=3$, $q=9$ and $k=8$, we have that $C(4)=\{4\}$, $C(5)=C(7)=\{ 5,7 \}$. By Proposition \ref{size}, we have 
    $$\dim_{\F_3} \ev{\Trqp}(\calM_4)=1, \quad \text{and} \quad \dim_{\F_3} \ev{\Trqp}(\calM_5)=\dim_{\F_3} \ev{\Trqp}(\calM_7)=2.$$
    Moreover, $\ev{\Trqp}(\calM_5) = \ev{\Trqp}(\calM_7)$ since
    $$\Tr_{\F_9/\F_3}(a\ell^5) = \Tr_{\F_9/\F_3}(a^3\ell^7),\ \forall a,\ell \in \F_9 \iff \Tr_{\F_9/\F_3}(ax^5) = \Tr_{\F_9/\F_3}(a^3x^7),\ \forall a \in \F_9.$$
\end{example}

To further characterize the  image $\ev{\Trqp}(\F_q[x]^{<k})$ we show that the images associated with distinct cyclotomic-coset representatives form a direct sum.    

\begin{proposition} \label{lin_indep}
    Let $L \subseteq \{ 0,\dots,k-1 \}$ such that $\{ C(l) \mid l \in L \}$ is a maximal size of distinct $p$-cyclotomic cosets modulo $(q-1)$. We have
    $$
    \ev{\Trqp}\left(\sum_{l \in L} \calM_l\right) = \bigoplus_{l \in L} \ev{\Trqp}(\calM_l).
    $$
\end{proposition}

\begin{proof}
    Since, for all $a \in \F_q$, $\Trqp(ax^l)$ is identically zero when evaluated at $0$ for all $l \in L \setminus\{ 0 \}$ and $\ev{\Trqp}(\calM_0)=\langle(1,\dots,1)\rangle_{\F_p}$, it holds that 
    $$
    \ev{\Trqp}\left(\sum_{l \in L} \calM_l\right) = \ev{\Trqp}(\calM_0) \oplus \ev{\Trqp}\left(\sum_{l \in L\setminus \{ 0 \}}\calM_l\right).
    $$ 
    Hence, with $L'=L \setminus\{ 0 \}$, it remains to prove that
    $$
    \ev{\Trqp}\left(\sum_{l \in L'} \calM_l\right) = \bigoplus_{l \in L'} \ev{\Trqp}(\calM_l).
    $$
    Let $f=\sum_{l\in L'}f_lx^l\in \sum_{l \in L'} \calM_l$, by the $\F_p$-linearity of the trace function, it holds that for any $b \in \F_q$
    \[\Trqp(f(b))=\Trqp\left(\sum_{l\in L'}f_lb^l\right)=\sum_{l\in L'}\Trqp(f_lb^l), \]
    implying the decomposition of the vector $\ev{\Trqp} (f)$ in sum of vectors of $\ev{\Trqp}(\calM_l)$.

    Now we need to prove that this decomposition is unique by showing that for all $j\in L'$,  \[\ev{\Trqp}(\calM_j) \cap \ev{\Trqp}\left(\sum_{l\in L'\setminus\{j\}}\calM_l\right)=\{0\}.\]
    
    For contradiction, suppose that the intersection is nontrivial for a $j \in L'$. This is equivalent to the existence of 
    \begin{equation}\label{eq:contradiction_intersection}\Trqp\left(a_l(\alpha^l)^i\right)\neq 0\end{equation} 
    for some $l\in L'$ and $i\in \{0,\dots,q-2\}$. Indeed, for the nontrivial intersection, there exists $a_j \in \F_q^*$ and $a_l \in \F_q$ such that
    $$\ev{\Trqp}(-a_jx^j) = \ev{\Trqp}\left(\sum_{l \in L'\setminus\{j\}} a_lx^l\right) = \sum_{l \in L'\setminus\{j\}} \ev{\Trqp}(a_lx^l)\neq 0.$$
    with some $\ev{\Trqp}(a_lx^l)$ not zero. Then 
    $$ \sum_{l \in L'\setminus\{j\}} \ev{\Trqp}(a_lx^l)-\ev{\Trqp}(-a_jx^j)=\sum_{l \in L'}\ev{\Trqp}(a_l x^l)=0,$$
    Looking at each coordinate, it holds that for $i\in \{0,\dots,q-2\}$
    \begin{equation} \label{sum_tr_zero}
        \sum_{l \in L'}\Trqp\left(a_l(\alpha^i)^l\right)=\sum_{l \in L'}\Trqp\left(a_l(\alpha^l)^i\right)=0,
    \end{equation}
    where some $\Trqp\left(a_l(\alpha^l)^i\right)\neq 0$. 
    
    Now, let $m_{\alpha^l} \in \F_p[x]$ be the minimal polynomial of $\alpha^l$ over $\F_p$, and 
    $$
    m_l = \prod_{j \in L' \setminus \{ l \}} m_{\alpha^j} \mod (x^{q-1}-1)
    \neq 0,$$ from the definition of $L'$.
    By defining 
    $$
    f_{l,s} = x^sm_l \mod (x^{q-1}-1)
    $$
    we have that
    \begin{equation}\label{eq:flsalphaj}f_{l,s}(\alpha^j) = 0, \ \forall j \in L' \setminus\{l\}.\end{equation} 
    
    Therefore, for $l \in L'$ and $s \in \{ 0,\dots,q-2 \}$, if $f_{l,s}=\sum_{t=0}^{q-2}f_tx^t$, by $\F_p$-linearity of $\Trqp$ and Equation \eqref{sum_tr_zero}, we have
    $$\sum_{j \in L'} \Trqp(a_j f_{l,s}(\alpha^j)) = \sum_{j \in L'} \Trqp\left(a_j \sum_{t=0}^{q-2} f_t (\alpha^j)^t\right) = \sum_{t=0}^{q-2} f_t \left(\sum_{j \in L'} \Trqp(a_j(\alpha^j)^i) \right)=0.$$
    So, we have
    $$0=\sum_{j \in L'} \Trqp(a_j f_{l,s}(\alpha^j))=\Trqp(a_lf_{l,s}(\alpha^l)) = \Trqp(a_lm_l(\alpha^l)(\alpha^l)^s)$$
    where the second equality follows from Equation \eqref{eq:flsalphaj}. 
    
    Since this holds for all $l\in L'$, $\Trqp(a_lm_l(\alpha^l)b)=0$ for $b\in\K_l$. Now, for all $l \in L'$, since $m_l(\alpha^l) \neq 0$ and $\K_l$ is a field, multiplying $\K_l$ by $m_l(\alpha^l)$ is a bijection into $\K_l$, implying that $\Trqp(a_lb)=0$ for $b\in \K_l$. By Proposition \ref{prop:tr_zero}, this implies that, for all $l \in L'$, $\Trqp(a_lb^l)=0$ for $b\in\F_q$, which contradicts Equation \eqref{sum_tr_zero}.
    
\end{proof}

\begin{example} \label{lin_size_example}
    Let us consider $p=3$, $q=3^2=9$ and $k=8$. We have already seen in Example \ref{cyclo_example} that $$C(0)=\{0\},  C(1)=C(3)=\{ 1,3 \},  C(2)=C(6)=\{ 2,6 \}, C(4)=\{4\}, \text{ and }  C(5)=C(7)=\{ 5,7 \}.$$
    So we have 
    $$L=\{ 0,1,2,4,5 \},$$
    and 
    $$\ev{\Trqp}(\calM_0 + \calM_1 + \calM_2 + \calM_4 + \calM_5) = \bigoplus_{l\in\{ 0,1,2,4,5 \}}\ev{\Trqp}(\calM_l).$$
\end{example}

A direct consequence of Proposition \ref{size} and Proposition \ref{lin_indep} is the following Corollary that states the cardinality of $\chi_p \circ \ev{\Trqp}(\F_q[x]^{<k})$.

\begin{corollary} \label{cor:size_image}
    If $L \subseteq \{ 0,\dots,k-1 \}$ such that $\{ C(l) \mid l \in L \}$ is a maximal size of distinct $p$-cyclotomic cosets modulo $(q-1)$, then 
    $$|\chi_p \circ \ev{\Trqp}(\F_q[x]^{<k})| = p^{\sum_{l \in L} |C(l)|}.$$
\end{corollary}

\begin{proof}
    Since $\chi_p:\F_p^q \to \mun{p}^q$ is injective, we have 
    $$
    |\chi_p \circ \ev{\Trqp}(\F_q[x]^{<k})| = |\ev{\Trqp}(\F_q[x]^{<k})|.
    $$
    However, by Proposition \ref{size} and Proposition \ref{lin_indep}, we have
    $$
    \dim_{\F_p} (\operatorname{Im} \ev{\Trqp}) = \sum_{l \in L} |C(l)| \implies |\ev{\Trqp}(\F_q[x]^{<k})|= p^{\sum_{l \in L} |C(l)|}.
    $$
\end{proof}

\begin{example}
    From Example \ref{lin_size_example}, we have 
    $$|C(0)|=1 \implies \dim_{\F_3} \ev{\Trqp}(\calM_0)=1,\quad |C(1)|=2 \implies \dim_{\F_3} \ev{\Trqp}(\calM_1)=2,$$ 
    $$|C(2)|=2 \implies \dim_{\F_3} \ev{\Trqp}(\calM_2)=2, \quad |C(4)|=1 \implies \dim_{\F_3} \ev{\Trqp}(\calM_4)=1,$$
    and 
    $$\quad |C(5)|=2 \implies \dim_{\F_3} \ev{\Trqp}(\calM_5)=2.$$
    So
    $$
    \dim_{\F_3} \operatorname{Im} \ev{\Trqp} = 1 +2 +2 +1+2 = 8.
    $$
\end{example}

\section{A Polynomial Family for the Trace Evaluation Map}\label{sec:trace-polynomials}
In this section, for $k\leq q$, we construct a set $\calM \subset \F_q[x]^{<k}$ for which 

\begin{itemize}
    \item $|\calM| = |\chi_p \circ \ev{\Trqp} (\F_q[x]^{<k})|$, and
    \item $\chi_p \circ \ev{\Trqp} (\calM)=\chi_p \circ \ev{\Trqp} (\F_q[x]^{<k})$.
\end{itemize}  

Since the extension of $\chi_p$ to $\F_p^q$ is bijective from $\F_p^q$ to $\mun{p}^q$, it is enough to  consider $\ev{\Trqp}$. 

It is not the first time that the trace function has been used in codes, notable examples are Trace codes contained in \cite{galindo2018classical} and \cite{guneri2024subfield}. Furthermore, evaluations followed by additive characters have also appeared in \cite{tang2017linear}.

It is also worth mentioning that for a Reed--Solomon code $\calC$ obtained by evaluating polynomials in
$\mathbb{F}_{q}[x]^{<k}$ at all elements of $\mathbb{F}_{q}$, the trace code
$\operatorname{Tr}_{\mathbb{F}_{q}/\mathbb{F}_{p}}(\calC)$, which is isometric with respect to the Hamming distance to $ev_{\Trqp}(\F_q[x]^{<k})$, admits a particularly
explicit description. In our setting, we determine its exact cardinality and
characterize the full set of polynomials $\mathcal{M} \subset \mathbb{F}_{q}[x]^{<k}$
for which $ev_{\Trqp}$ induce a one-to-one
correspondence between $\mathcal{M}$ and $ev_{\Trqp}(\F_q[x]^{<k})$.
Although this is not the main focus of the paper, these observations are
conceptually useful: combined with Delsarte's theorem (Theorem 2 \cite{delsarte1975subfield}),
\[
C_{\mid \mathbb{F}_{p}}
    = \bigl(\operatorname{Tr}_{\mathbb{F}_{q}/\mathbb{F}_{p}}(C^{\perp})\bigr)^{\perp},
\]
they provide an explicit polynomial-level understanding of the trace code
appearing in the characterization of the subfield subcode $\calC_{\mid \mathbb{F}_{p}}$. 

We first record a useful consequence of the structure of $\K_l$ for $l \in L$ where $L \subseteq \{ 0,\dots,k-1 \}$ such that $\{ C(l) \mid l \in L \}$ is a maximal size of distinct $p$-cyclotomic cosets modulo $(q-1)$ such as in Section \ref{sec:cardinality}.

We also need the following lemma that fully characterizes $\operatorname{ker}(\ev{\Trqp}|_{ \calM_l})$.

\begin{lemma}\label{lemma:ker(evtr)}
    For $l \in L$, we have that
     $$\left \{ cx^l \mid c \in \operatorname{ker}(\Tr_{\F_q/\K_l}) \right \} = \operatorname{ker}(\ev{\Trqp}|_{ \calM_l}).$$
\end{lemma}

\begin{proof}
    Let  $a \in \operatorname{ker}(\Tr_{\F_q/\K_l})$. For $b \in \F_q$, since $\Trqp=\Tr_{\K_l/\F_p} \circ \Tr_{\F_q/\K_l}$, we have
$$\Trqp(ab^l)=\Tr_{\K_l/\F_p}(\Tr_{\F_q/\K_l}(ab^l))=\Tr_{\K_l/\F_p}(b^l\Tr_{\F_q/\K_l}(a))=0,$$
where the second equality comes from the fact that $b^l \in \K_l$. Therefore, we have
$$\ev{\Trqp}(ax^l)=0.$$
This implies that
$$\left \{ cx^l \mid c \in \operatorname{ker}(\Tr_{\F_q/\K_l}) \right \} \subseteq \operatorname{ker}(\ev{\Trqp}|_{ \calM_l}).$$

For the equality, let us prove that these sets have the same cardinality. By the rank-nullity theorem and Proposition \ref{size}, we have
$$\dim_{\F_p}\operatorname{ker}(\ev{\Trqp}|_{\calM_l})=r-|C(l)| \implies|\operatorname{ker}(\ev{\Trqp}|_{\calM_l})| = p^{r-|C(l)|}.$$
    Since $|C(l)|$ divides $r$, let $d$ such that $r=|C(l)|d$. It holds that
    $$|\operatorname{ker}(\ev{\Trqp}|_{\calM_l})|=p^{|C(l)|d-|C(l)|}=(p^{|C(l)|})^{d-1}=|\K_l|^{d-1} =  |\operatorname{ker}(\Tr_{\F_q/\K_l})|,$$
    where the last equality comes from the rank-nullity Theorem on $\Tr_{\F_q/\K_l}:\F_q \to \K_l$. We then have 
    $$|\operatorname{ker}(\ev{\Trqp}|_{\calM_l})|=|\operatorname{ker}(\Tr_{\F_q/\K_l})|=\left| \left \{ bx^l \mid b \in \operatorname{ker}(\Tr_{\F_q/\K_l}) \right \} \right|.$$
\end{proof}

Now notice that from Example \ref{ex:overest_M_4}, $\calM_5$ maps injectively into $\ev{\Trqp}(\calM_5)$, and the same holds for $\calM_7$ and $\ev{\Trqp}(\calM_7)$; however, this is not the case for $\calM_4$ and $\ev{\Trqp}(\calM_4)$ since \[\dim_{\F_3}\ev{\Trqp}(\calM_4)=1<2=\dim_{\F_3}\calM_4.\] 
For $l \in L$, we give a subset of $\calM_l$ for which $\ev{\Trqp}$ restricted to the subset is injective.

\begin{proposition}\label{CP_refinement}
    Let $L \subseteq \{ 0,\dots,k-1 \}$ such that $\{ C(l) \mid l \in L \}$ is a maximal size of distinct $p$-cyclotomic cosets modulo $(q-1)$ and let $z\in \F_q$ such that $\Tr_{\F_q/\K_l}(z) \neq 0$. If
        $$
        \calM_l'=\{ ax^l \mid a \in \langle z \rangle_{\K_l} \},$$ 
        then it holds that
        \begin{itemize}
            \item $|\calM_l'|=|\ev{\Trqp}(\calM_l)| $, and
            \item $\ev{\Trqp}(\calM_l')= \ev{\Trqp}(\calM_l)$. 
        \end{itemize}
\end{proposition}

\begin{proof}
It holds that 
$$|\calM_l'| = |\K_l| = p^{|C(l)|} = |\ev{\Trqp} (\calM_l)|,$$
where the last equality comes from Proposition \ref{size}. 

To prove $\ev{\Trqp}(\calM_l')= \ev{\Trqp}(\calM_l)$ it  suffices to prove that $\ev{\Trqp}$ restricted to ${\calM_l'}$ is injective by proving that $\operatorname{ker}(\ev{\Trqp}|_{\calM'_l})=\{0\}$. For this, let $ax^l \in \operatorname{ker}(\ev{\Trqp}|_{\calM'_l})$. We prove that $a=0$. By Lemma \ref{lemma:ker(evtr)}, we have $a \in \operatorname{ker}(\Tr_{\F_q/\K_l}).$
Moreover, since $a \in \langle z \rangle_{\K_l}$, there exists $b \in \K_l$ such that $a=bz$. Given that  $\Tr_{\F_q/\K_l}(z)\neq 0$, we have
$$0=\Tr_{\F_q/\K_l}(a)=\Tr_{\F_q/\K_l}(bz)=b\Tr_{\F_q/\K_l}(z) \implies b=0,$$ 
and therefore, $a=0$.
\end{proof}

\begin{example}
    Let us consider $p=3$ and $q=9$. 
    \begin{itemize}
        \item For $l=0$, we have $\K_0 = \langle 1 \rangle_{\F_3} = \F_3.$
        Since $\Tr_{\F_9/\F_3}(1) \neq 0$, we can consider $z=1$ and
        $$ \calM_0'=\langle 1 \rangle_{\F_3} =\F_3 .$$
        \item For $l=2$ we have 
        $$\gcd(l,q-1)=2 \implies \operatorname{ord}(\alpha^l)=\frac{8}{2}=4,$$
        and since $s=2$ is minimal such that $4 \mid (3^s-1)$ we have
        $$\K_2 = \F_{3^2}=\F_9.$$
        Therefore $\operatorname{ker}(\Tr_{F_9/\K_2})$ is trivial. So 
        $$
        \calM_2'=\{ ax^2 \mid a \in \langle 1 \rangle_{\F_9} = \F_9 \}=\calM_2.
        $$
        \item For $l=5$ we have
        $$\gcd(l,q-1)=1 \implies \operatorname{ord}(\alpha^5)=8 \implies \K_5=\F_9.$$
        Therefore $\operatorname{ker}(\Tr_{F_9/\K_5})$ is trivial. So
        $$\calM_5' = \{ ax^5 \mid a \in \langle 1 \rangle_{\F_9} = \F_9 \}=\calM_5.$$
        \item For $l=4$ we have
        $$\gcd(l,q-1)=4 \implies \operatorname{ord}(\alpha^l)=\frac{8}{4}=2,$$
        and since $s=1$ is minimal such that $2 \mid (3^s-1)$, we have
        $$\K_4 = \F_{3^1}=\F_3.$$
        With $\Tr_{\F_9/\F_3}(1) \neq 0$, we can consider $z=1$. Thus
        $$
        \calM_4' = \{ ax^4 \mid a \in \F_3 \}
        $$
        \item For $l=8$, we have
        $$\K_8 = \langle 1 \rangle_{\F_3} = \F_3.$$
        Since $\Tr_{\F_9/\F_3} (1) \neq 0$, we can consider $z=1$ and 
        $$\calM_8' = \{ ax^8 \mid a \in \F_3\}$$
    \end{itemize}
\end{example}

The following corollary constructs $\calM$ and is a direct consequence of Corollary \ref{cor:size_image}, Proposition \ref{lin_indep}, Proposition \ref{CP_refinement} and the fact that $\chi_p:\F_p^q \to \mun{p}^q$ is bijective. 

\begin{corollary} \label{cor:M_prop}
    Let $L \subseteq \{ 0,\dots,k-1 \}$ such that $\{ C(l) \mid l \in L \}$ is a maximal size of distinct $p$-cyclotomic cosets modulo $(q-1)$ and let $\calM_l'$ as in Proposition \ref{CP_refinement}. With
    $\calM = \bigoplus_{l \in L} \calM_l',$
    the following holds
    \begin{enumerate}
        \item $|\calM| = |\chi_p \circ \ev{\Trqp} (\F_q[x]^{<k})|$ 
        \item $\chi_p \circ \ev{\Trqp} (\calM)=\chi_p \circ \ev{\Trqp} (\F_q[x]^{<k})$.
    \end{enumerate}
\end{corollary}

\begin{proof}
    We have
    $$|\chi_p \circ \ev{\Trqp} (\F_q[x]^{<k})|=\left |\ev{\Trqp} (\F_q[x]^{<k}) \right|=\left |\ev{\Trqp} \left(\sum_{l \in L} \calM_l\right) \right|,$$
    where the last equality comes from  Corollary \ref{cor:size_image}. However, by Proposition \ref{lin_indep}, we have
    $$\left |\ev{\Trqp} \left(\sum_{l \in L} \calM_l\right) \right| = \prod_{l \in L}|\ev{\Trqp}\left(\calM_l\right)|.$$
    By Proposition \ref{CP_refinement}, we have 
    $$\prod_{l \in L}|\ev{\Trqp}(\calM_l)|=\prod_{l \in L}|(\calM_l')|=|\calM|.$$

    For the second part of the corollary, from Proposition \ref{size} and Proposition \ref{lin_indep}, we have
    $$\ev{\Trqp} (\F_q[x]^{<k})=\ev{\Trqp}\left(\sum_{l \in L}\calM_l\right) = \bigoplus_{l \in L} \ev{\Trqp}(\calM_l).$$
    By Proposition \ref{CP_refinement}, we have
    $$\bigoplus_{l \in L} \ev{\Trqp}(\calM_l)=\bigoplus_{l \in L} \ev{\Trqp}(\calM_l')=\ev{\Trqp}(\calM).$$
    Therefore, since $\chi_p$ is bijective, we have
    $$\chi_p \circ \ev{\Trqp} (\calM)=\chi_p \circ \ev{\Trqp} (\F_q[x]^{<k}).$$
    
\end{proof}

\section{Cardinality and Polynomial Family for CP codes}\label{sec:CP-polynomials}
Fix $p,r$ and $k$. Let $L \subseteq \{ 0,\dots,k-1 \}$ such that $\{ C(l) \mid l \in L \}$ is a maximal size of distinct $p$-cyclotomic cosets modulo $(q-1)$, and $\calM = \bigoplus_{l \in L} \calM_l'$ as in Section \ref{sec:trace-polynomials}. In this section, from a specific choice of $L$, we construct, from $\calM$, a set $\calF' \subset \calF$ for which $\calC(\calF') = \calC(\calF)$. With this, we then obtain $|\calC(\calF)|$, and, with Proposition \ref{CP_refinement}, we prove that the choice is actually independent of the choice of $L$, which need not be
unique. For instance, with $p=3,q=9$ and $k=8$, one can consider $L=\{ 0,1,2,4,7\}$ instead of $L=\{0,1,2,4,5\}$.

\begin{theorem}\label{CP:size&poly}
    Let $p=q^r$, $\chi$ a nontrivial additive character of $\F_q$, $k \in \N$ with $k \leq q$ and $m=q$. With $\calF$ and $\calC(\calF)$ as in Definition \ref{def:cp_code} and $L \subseteq \{ 0,\dots,k-1 \}$ such that $\{ C(l) \mid l \in L \}$ is a maximal size of distinct $p$-cyclotomic cosets modulo $(q-1)$, we have 
    $$|\calC(\calF)| = p^{\dim_{\F_p} \ev{\Trqp}(\F_q[x]^{<k})-1}=p^{\sum_{l \in L}|C(l)|-1}.$$
    Moreover, if $\calF' = \bigoplus_{l \in L \setminus\{0\}} \calM_l'$ where $\calM_l'$ be as in Proposition \ref{CP_refinement}, and $\langle \cdot \rangle_{\C}:\C^m \to G_{1,m}(\C^m)$, we have
    $$
    \langle \cdot \rangle_\C \circ\chi_p \circ \ev{\Trqp}(\calF')=\calC(\calF).
    $$
\end{theorem}

\begin{proof}
    Since $\chi_p$ is bijective and $\ev{\Trqp}(\calM) = \ev{\Trqp}(\F_q[x]^{<k})$, it suffices to consider polynomials in $\calM$. Moreover, since $\calC(\calF)$ is a collection of 1-dimensional vector spaces, it is enough to determine when 
    $\chi_p \circ \ev{\Trqp} (f)= (\chi(f(a_1)),\dots,\chi(f(a_m)))$ and $\chi_p \circ \ev{\Trqp} (g)=(\chi(g(a_1)),\dots,\chi(g(a_m)))$ are scalar multiple of each other for distinct polynomials $f,g\in \calM$. That is, suppose there exists $\lambda \in \mun{p}$ such that 
    $$\chi(f(a)) = \lambda \chi(g(a)),\ \forall a \in \F_q.$$
    Since $\chi_p : \F_p \to \mun{p}$ is bijective, there exists a unique $c \in \F_p$ such that $\lambda=\chi_p(c)$ and
    $$\chi_p \circ \Trqp(f(a)) = \chi_p(c)\chi_p\left(\Trqp(g(a)) \right) = \chi_p(c + \Trqp(g(a))),\ \forall a \in \F_q.$$
    Therefore, we have
    $$\Trqp(f(a)) = c + \Trqp(g(a)),\ \forall a \in \F_q \iff \Trqp\left( (f-g)(a) \right) = c,\  \forall a\in \F_q.$$
    That is 
    $$\ev{\Trqp}(f-g) \in \ev{\Trqp}(\calM_0')$$ 
    and, since $\ev{\Trqp}|_{\calM}$ is bijective, we have 
    $$f-g \in \calM_0'.$$
    Therefore, we have 
    $$f-g \in \calM_0' \iff \exists \lambda\in \mun{p}\ \text{ such that }\ \chi(f(a)) = \lambda \chi(g(a)),\ \forall a \in \F_q,$$
    meaning that $\langle \cdot \rangle_{\C} \circ\chi_p \circ \ev{\Trqp}$ restricted to $\calM/\calM_0'$ is injective. Hence
    $$
    \langle \cdot \rangle_{\C} \circ\chi_p \circ \ev{\Trqp}(\calM) = \langle \cdot \rangle_{\C} \circ \chi_p \circ \ev{\Trqp}(\calM/\calM_0') = \langle \cdot \rangle_{\C} \circ \chi_p \circ \ev{\Trqp}(\calF').
    $$
    By Corollary \ref{cor:M_prop} we have
    $$\langle \cdot \rangle_{\C} \circ \chi_p \circ \ev{\Trqp}(\F_q[x]^{<k})=\langle \cdot \rangle_{\C} \circ \chi_p \circ \ev{\Trqp}(\calM)=\langle \cdot \rangle_{\C} \circ \chi_p \circ \ev{\Trqp}(\calF').$$
    Now, given the definition of the $p$-cyclotomic cosets modulo $(q-1)$, one can choose every element of $L$ to be minimal so that $\gcd(l,p)=1$ for all $l \in L \setminus \{0\}$. In this situation, it then holds that 
    $$
    \calF'\subseteq\calF \subseteq \F_q[x]^{<k} \implies \langle \cdot \rangle_{\C} \circ\chi_p \circ \ev{\Trqp}(\calF') \subseteq \langle \cdot \rangle_{\C} \circ\chi_p \circ \ev{\Trqp}(\calF) \subseteq \langle \cdot \rangle_{\C} \circ \chi_p \circ \ev{\Trqp}(\F_q[x]^{<k}).
    $$
    Hence,
    $$ 
    \langle \cdot \rangle_{\C} \circ \chi_p \circ \ev{\Trqp}(\calF') = \langle \cdot \rangle_{\C} \circ \chi_p \circ \ev{\Trqp}(\calF) =\calC(\calF),
    $$
    and
    $$|C(\calF)| = |\langle \cdot \rangle_{\C} \circ\chi_p \circ \ev{\Trqp}(\calF')| = \left| \calF' \right| = p^{\sum_{l \in L\setminus\{0\}}|C(l)|}=p^{\sum_{l \in L} |C(l)|-1}.$$
    The proof then ends with the fact that $\ev{\Trqp}(\calF')$ is independent of the choice of $L$ and $0 \in L$.

\end{proof}

\begin{example}
    Let us consider $p=3,q=9$ and $k=8$. We have 
    $$\calM_1'=\{ ax \mid a \in \F_9 \}, \quad \calM_2'=\{ ax^2 \in \F_9 \}, \quad \calM_4' = \{ax^4 \mid a \in \F_3 \}, \quad{and} \quad \calM_5'=\{ ax^5 \mid a \in \F_9 \}$$
    So 
    $$
    \calF' = \calM_1' \oplus \calM_2' \oplus \calM_4' \oplus \calM_5', \quad \text{and} \quad |\calF'|=3^7=2187
    $$
    and
    $$\calC(\calF) = \langle \cdot \rangle_{\C} \circ\chi_p \circ \ev{\Trqp}(\calF')$$
\end{example}

\section*{Conclusion}

This paper gives an algebraic characterization of the redundancies in character-polynomial
(CP) codes. Motivated by explicit examples and simulations showing that the defining
polynomial family can be larger than the resulting set of codewords, we factor the
construction into a trace-evaluation map followed by an additive character map and
projectivization. This decomposition identifies the source of the cardinality loss and makes it possible to count the distinct codewords exactly.

A central component of our analysis is the use of $p$-cyclotomic cosets modulo $(q-1)$ and their connection to finite field extensions. This structure allowed us to determine the exact number of distinct trace-evaluation vectors produced by polynomials of degree less than $k$. From this characterization, we identified a subset $\mathcal{M} \subset \mathbb{F}_q[x]^{<k}$ that maps bijectively to the trace-evaluation image, thereby isolating the essential polynomial representatives.

Building on this foundation, we constructed the refined polynomial family $\mathcal{F}'$. This set is a central contribution of the paper: $\mathcal{F}'$ is precisely the collection of polynomials that generate the CP code without any redundancy, and each element of $\mathcal{F}'$ corresponds uniquely to a codeword in $\mathcal{C}(\mathcal{F})$. As a consequence, the cardinality of the CP code is exactly $|\mathcal{F}'|$, and our construction provides an explicit algebraic formula for this quantity. Thus, beyond determining the size of the CP code, we also identify the exact polynomial set that should be used in its construction. 

Together, these results yield a complete and minimal description of CP codes, clarifying both their internal algebraic structure and the precise relationship between polynomial families and codewords. The identification of $\mathcal{F}'$ as the minimal generating set provides a foundation for more efficient implementations, sharper parameter analyses, and potential extensions to broader classes of trace-based codes.

\bibliographystyle{plain}
\bibliography{biblio}

\end{document}